\documentclass[onecolumn,draftclsnofoot,12pt]{IEEEtran}
\IEEEoverridecommandlockouts
\usepackage{tikz}
\usetikzlibrary{calc}
\usepackage{cite}
\usepackage[normalem]{ulem}
\usepackage{amsmath,amssymb,amsfonts,amsthm,verbatim,mathtools}
\usepackage{algorithmic}
\usepackage{graphicx}
\usepackage{textcomp}
\usepackage{xcolor}
\usepackage{pgfplots}
\pgfplotsset{compat=1.7,
emphasize/.code args={#1:#2with#3}{
    \pgfplotsextra{
            \draw[fill=#3] ({axis cs:#1,0} |- {axis description cs:0,0}) 
            rectangle ({axis cs:#2,0} |- {axis description cs:0,1});
        }
    }
}
\usepackage[utf8]{inputenc}
\usepackage[T1]{fontenc}

\def\BibTeX{{\rm B\kern-.05em{\sc i\kern-.025em b}\kern-.08em
    T\kern-.1667em\lower.7ex\hbox{E}\kern-.125emX}}

\allowdisplaybreaks

\newtheorem{definition}{Definition}
\newtheorem{lemma}{Lemma}
\newtheorem{proposition}{Proposition}
\newtheorem{theorem}{Theorem}

\newtheorem{example}{Example}
\newtheorem{corollary}{Corollary}

\renewcommand{\baselinestretch}{0.97}

\begin{document}

\title{Secret Sharing in the Rank Metric}


 \author{\small
   \IEEEauthorblockN{Johan Vester Dinesen\IEEEauthorrefmark{1}\thanks{This publication has emanated from research conducted with the financial support of the European Union MSCA Doctoral Networks, (HORIZON-MSCA-2021-DN-01, Project 101072316).
},
                     Eimear Byrne\IEEEauthorrefmark{2},
                     Ragnar Freij-Hollanti\IEEEauthorrefmark{1},
                     and Camilla Hollanti\IEEEauthorrefmark{1}} \\
              \IEEEauthorblockA{\IEEEauthorrefmark{1}%
                    Department of Mathematics and Systems Analysis, Aalto University, Finland,
                first.last@aalto.fi} \\
                \IEEEauthorblockA{\IEEEauthorrefmark{2}%
                     School of Mathematics and Statistics, University College Dublin, Ireland, 
            ebyrne@ucd.ie}
 }

\maketitle

\begin{abstract}
The connection between secret sharing and matroid theory is well established. In this paper, we generalize the concepts of secret sharing and matroid ports to $q$-polymatroids. Specifically, we introduce the notion of an access structure on a vector space, and consider properties related to duality, minors, and the relationship to $q$-polymatroids. Finally, we show how rank-metric codes give rise to secret sharing schemes within this framework.

\end{abstract}

\section{Introduction}

Secret sharing was independently introduced by Blakley \cite{Blakley} and Shamir \cite{shamir79} in 1979. It is a cryptographic primitive where a secret is divided into multiple shares, which are then distributed among participants such that only certain authorized collections of participants can reconstruct the secret.  One way to construct secret sharing schemes is via classical linear error-correcting codes. 
In particular, Massey demonstrated that the minimal codewords of the dual code determine the access structure of such schemes \cite{Massey1999MinimalCA}. Lehman introduced matroid ports in 1964 \cite{Lehman} and their connection to secret sharing is well established \cite{farre_padro2007}. Since linear codes induce matroids \cite{Oxley} this further establishes the link between secret sharing and linear codes. Matroids and polymatroids also have an abundance of applications in distributed data storage \cite{westerback-book}.

Linear rank-metric codes were originally introduced by Delsarte \cite{DELSARTE1978226} as sets of matrices over a finite field $\mathbb{F}_q$, where the distance between two matrices is defined as the rank of their difference. These codes were later rediscovered by Roth \cite{Roth} and Gabidulin \cite{Gabidulin}. Their properties have been widely studied in the context of random linear network coding \cite{4567581, 4608992}, {including coset coding to secure the network from wiretappers \cite{Universal,Emina}. In addition, rank-metric codes have been proposed for post-quantum encryption and digital signature schemes. 
  Notably, ROLLO and RQC are examples of rank-metric-based cryptosystem submissions to the NIST Post-Quantum Cryptography Standardization process \cite{Rollo, RQC}. However, ROLLO and RQC did not advance in the PQC standardization process, as their security analyses required further scrutiny. Despite this, NIST recognizes that rank-based cryptography remains competitive and recommends continued research in this area \cite{257451}. As a result, rank-metric cryptanalysis remains an active area of research; see \cite{cryptoeprint:2022/1031, 10.1109/TC.2022.3225080,10.1007/978-3-030-98365-91,RYDE}. Lastly, $q$-polymatroids, which generalize (poly)matroids, are well known to arise from rank-metric codes \cite{Gorla_2019, Johnsen_2020, Shiromoto}. Consequently, $q$-polymatroids serve as a valuable tool to gain insights into rank-metric codes.

In this paper, we consider secret sharing in the rank-metric, one application of which is the mentioned wiretap network coding. We generalize access structures on finite sets to vector spaces and extend the concept of matroid ports to that of $q$-polymatroid ports. We establish notions of duality and minors of access structures on vector spaces, and demonstrate their relationships to the duality and minors of $q$-polymatroids and $q$-polymatroid ports. Furthermore, we explore perfect schemes, minimum gaps, and threshold structures within this framework. We also investigate how rank-metric codes and their associated $q$-polymatroids give rise to secret sharing schemes. In particular, we show that maximum rank distance (MRD) rank-metric codes yield perfect threshold schemes within the framework presented in this paper. Finally, we demonstrate that the rank function of the $q$-polymatroid induced by a rank-metric code can be determined using the entropy of random variables associated with the same code. 
\section{Preliminaries}

In this section we briefly summarize well-known definitions and results needed in the paper.

\subsection{Access Structures and Secret Sharing}

A general way to define linear secret sharing \cite{shamir79} is via (multi-target) monotone span programs, see \cite{beimel2011secret} for more details. In terms of linear codes, secret sharing can be described by nested coset coding similarly to wiretap coding, see \cite{chen2007secure, duursma2010coset, kurihara2012secret, martinez2018communication}. Secret sharing is also a crucial component in private information retrieval \cite{chor1995private,Sun2016,Banawan2018,freij2017private}, see \cite{song2022equivalence,okkoAG}. 
For a general survey on secret sharing schemes, we refer to \cite{beimel2011secret}.

In a secret sharing scheme, a dealer distributes shares to players such that only certain collections of players are able to reconstruct the secret. This set of collections is called the \emph{reconstructing structure}. Conversely, certain collections of players gain zero information about the secret. This set is called the \emph{privacy structure}. 

Next, let us briefly describe how Reed--Solomon codes yield secret sharing schemes. Let $0 \leq k \leq n \leq q$ be integers and denote the vector space of polynomials over a finite field $\mathbb{F}_q$ with degree at most $k-1$ by  $\mathbb{F}_q[X]_{<k}$.  A Reed--Solomon code can be defined as 
\[
\mathrm{RS}_{n,k}(\alpha) = \{ (f(\alpha_0),\ldots,f(\alpha_{n-1})) \,|\, f\in \mathbb{F}_q[X]_{<k} \},
\]
where $\alpha = (\alpha_0,\ldots,\alpha_{n-1}) \in \mathbb{F}_q^n$ and $\alpha_i \neq \alpha_j$ when $i\neq j$. Assume $\alpha_0 = 0$, and let $P_0 = \{0\}$ denote the dealer and $P = \{1, \ldots, n-1\}$ denote the set of players. A secret sharing scheme can be constructed as follows: The dealer chooses a secret $s_0 \in \mathbb{F}_q$ to share and uniformly at random chooses a codeword among $c = (s_0,c_1,\ldots,c_{n-1})\in \mathrm{RS}_{n,k}(\alpha)$. They then transmit $c_i$ to player $i \in P$. The secret can then be reconstructed by polynomial interpolation. The reconstructing structure then consists of any $B\subseteq P$ with $|B|\geq k$. The privacy structure consists of $B\subseteq P$ with $|B|\leq k-1$. Thus, a coalition $B\subseteq P$ can reconstruct $s_0$ if and only if it contains an information set of $\mathrm{RS}_{n,k}(\alpha)$.

We now consider access structures on finite sets and examine how these relate to secret sharing schemes; see \cite{cryptoeprint:2009/141,Cramer_Damgaard_Nielsen_2015}. Let $P$ be a finite set, and let $G, A \subseteq 2^P$. If $V \in G$ and $V' \subseteq P$ with $V \subseteq V'$ implies $V' \in G$, we call $G$ a \emph{monotone} set. Dually, if $W \in A$ and $W' \subseteq P$ with $W' \subseteq W$ implies $W' \in A$, we call $A$ an \emph{anti-monotone} set. These objects are also referred to in the literature as up-sets and down-sets, or as filters and ideals\cite{blyth2005lattices}.
\begin{definition}
Let $P$ be a finite set. Let $G, A\subseteq 2^P$ such that $G \cap A=\varnothing$, $G$ is monotone, and $A$ is an anti-monotone. Then $(G,A)$ is an \emph{access structure} on $P$.
\end{definition}

Clearly any reconstructing structure of a secret sharing scheme should be a monotone set in the player set, while the privacy structure is anti-monotone. An access structure $(G,A)$ is a $k$\emph{-threshold structure} if $G = \{B \subseteq P \,|\, k \leq |B|\}$. It is \emph{perfect} if $G \cup A=2^P$.  The previous example of a secret sharing scheme induced by the Reed--Solomon code therefore yields a perfect $k$-threshold structure.

\subsection{$q$-Polymatroids}
We now introduce the necessary definitions and results on $q$-polymatroids for this paper. For an overview of $q$-polymatroids, see \cite{DBLP:journals/corr/abs-2104-06570, BYRNE2024105799}.

We first establish some notation. We let $q$ denote a fixed prime power and let $\mathbb{F}_q$ be the finite field of order $q$. For any vector space $A$ over $\mathbb{F}_q$, let 
$\mathcal{L}(A)$ denote the poset of subspaces of $A$, ordered with respect to subspace-inclusion $\leq$. We let $E$ denote a fixed $n$-dimensional vector space over $\mathbb{F}_q$. For any $i\in [n]$,  $\mathbf{e}_i \in \mathbb{F}_q^n$ denotes the $i$'th standard basis vector.

For any vector space $A$ over $\mathbb{F}_q$ and $V\in \mathcal{L}(A)$,  $V^{\perp A}$ denotes the orthogonal complement of $V$ with respect to some fixed non-degenerate symmetric bilinear form. That is, if $\langle \, ,\, \rangle \colon A\times A \rightarrow \mathbb{F}_q$ is a fixed non-degenerate symmetric bilinear form then 
\[
    V^{\perp A} = \{w\in A \,|\, \langle v,w \rangle = 0 \text{ for all } w\in A \}.
\]
We write $V^\bot$ to denote $V^{\bot E}$.
\begin{definition} A $q$-polymatroid is a pair $\mathcal{M} = (E,\rho)$ where $\rho\colon \mathcal{L} (E) \rightarrow \mathbb{Q}_{\geq 0}$ is a rank function such that
\begin{enumerate}
    \item[(R1)] Boundedness: $0\leq \rho(V) \leq \dim V$ for all $V\in \mathcal{L}(E)$.
    \item[(R2)] Monotonicity: $\rho(V)\leq \rho(W)$ for all $V,W\in \mathcal{L}(E)$ with $V\leq W$.
    \item[(R3)] Submodularity: $\rho(V + W) + \rho(V\cap W ) \leq \rho(V) + \rho(W)$ for all $V,W\in \mathcal{L}(E)$.
\end{enumerate}
\end{definition}

A $q$-polymatroid such that $\rho\colon \mathcal{L}(E)\rightarrow \mathbb{N}_0$ is called a $q$-matroid. Furthermore, for any $q$-polymatroid $\mathcal{M} = (E,\rho)$ and for any $V,W\in \mathcal{L}(E)$ define $$\rho(V\,|\,W) \coloneqq \rho(V+W) - \rho(V)$$ as the conditional rank of $V$ given $W$.

There are many cryptomorphic definitions for $q$-matroids; however, it remains an open research problem whether these can all be generalized to $q$-polymatroids, see \cite{BYRNE2022149}. As such, the following definitions are provided specifically for $q$-matroids.
\begin{definition} Let $\mathcal{M} = (E,\rho)$ be a $q$-matroid and $V\in\mathcal{L}(E)$. 
\begin{enumerate}
    \item $V$ is an independent space if $\rho(V) = \dim V$.
    \item $V$ is a dependent space if it is not independent.
    \item $V$ is a circuit if it is dependent and all proper subspaces of $V$ are independent.
    \item $V$ is a basis if it is independent and is not properly contained in any independent spaces.
\end{enumerate}
\end{definition}

Lastly, we define duality, equivalence, and minors of $q$-polymatroids.

\begin{definition}
    Let $\mathcal{M} = (E,\rho)$ be a $q$-polymatroid. Define $\rho^*(V) \coloneqq \dim V - \rho(E) + \rho(V^{\perp})$ for all $V\in \mathcal{L}(E)$. Then $\mathcal{M}^* \coloneqq (E,\rho^*)$ is a $q$-polymatroid called the dual of $\mathcal{M}$, and $(\mathcal{M}^*)^* = \mathcal{M}$.
\end{definition}

\begin{definition} \label{def:equivalentmatroids}
Let $\mathcal{M}_i = (E_i, \rho_i)$ be $q$-polymatroids for $i=1,2$. Then $\mathcal{M}_1$ and $\mathcal{M}_2$ are equivalent, denoted $\mathcal{M}_1 \simeq \mathcal{M}_2$, if there exists an $\mathbb{F}_q$-linear  isomorphism $\phi \colon \mathcal{L}(E_1) \rightarrow \mathcal{L}(E_2)$ such that $\rho_2(\phi(V)) = \rho_1(V)$ for all $V\in \mathcal{L}(E_1)$. 
\end{definition}

\begin{definition} Let $\mathcal{M} = (E,\rho)$ be a $q$-polymatroid and $Z\in \mathcal{L}(E)$. Then $\mathcal{M}|_Z \coloneqq (Z,\rho|_{Z})$ where $\rho|_{Z}(V) \coloneqq \rho(V)$ for all $V\in \mathcal{L}(Z)$ is a $q$-polymatroid, called the restriction of $\mathcal{M}$ to $Z$. Additionally, let $\pi\colon E\rightarrow E/Z$ denote the canonical projection. Then $\mathcal{M}/Z \coloneqq (E/Z,\rho_{E/Z})$ where $\rho_{E/Z}(V) \coloneqq \rho(\pi^{-1}(V))-\rho(Z)$ is a $q$-polymatroid, called the contraction of $Z$ from $\mathcal{M}$.
\end{definition}

\subsection{Rank-Metric Codes}
Let us recall some standard definitions and properties of rank-metric codes. For an overview of rank-metric codes, we refer to \cite{bartz2022rankmetriccodesapplications}. 

A \emph{rank-metric code} is an $\mathbb{F}_q$-linear subspace $\mathcal{C}\in \mathcal{L}(\mathbb{F}_q^{n\times m})$ equipped with the rank distance $d_{\mathrm{rk}}(X,Y) = \mathrm{rk}(X-Y)$ for any $X,Y\in \mathcal{C}$.
         The dual code of $\mathcal{C}$ is $$\mathcal{C}^\perp = \{X\in \mathbb{F}_q^{n\times m} \,|\, \mathrm{tr}(XY^\top) = 0 \text{ for all } Y\in \mathcal{C}\},$$ where $\mathrm{tr}$ is the trace of a matrix.
         For any rank-metric code $\mathcal{C}$ we define the minimum rank distance as         
     $$d_{\mathrm{rk}}(\mathcal{C}) = \min\{\mathrm{rk}(X) \,|\,X\in \mathcal{C}, X\neq 0\},$$  and $\mathcal{C}$ then satisfies the rank-metric Singleton bound $$\dim \mathcal{C}\leq \max\{m,n\}(\min\{m,n\}-d_{\mathrm{rk}}(\mathcal{C}) +1).$$ If $\mathcal{C}$ meets the Singleton bound with equality it is called a maximum rank distance (MRD) code. If $\mathcal{C}$ is MRD then $\mathcal{C}^\perp$ is MRD with $d_{\mathrm{rk}}(\mathcal{C}^\perp) = \min\{m,n\}-d_{\mathrm{rk}}(\mathcal{C})+2$.
         Let $V\leq \mathbb{F}_q^n$. The \emph{shortening of $\mathcal{C}$ with respect to $V$} is $$\mathcal{C}(V) = \{X \in \mathcal{C}\,|\, \mathrm{colsp}(X) \leq V\}.$$
         The induced $q$-polymatroid of $\mathcal{C}$ is $\mathcal{M}_{\mathcal{C}} = (\mathbb{F}_q^n, \rho_{\mathcal{C}})$ where $$\rho_{\mathcal{C}}(V) \coloneqq (\dim \mathcal{C} - \dim \mathcal{C}(V^{\perp}))/m$$ for any $V\in \mathcal{L}(\mathbb{F}_q^n)$ and it satisfies $\mathcal{M}_{\mathcal{C}}^*= \mathcal{M}_{\mathcal{C}^\perp}$. Let $\mathcal{M} = (E,\rho)$ be a $q$-polymatroid, where $E$ is an $n$-dimensional $\mathbb{F}_q$-vector space. Then $\mathcal{M}$ is \emph{representable} if there exists a rank-metric code $\mathcal{C}\in \mathcal{L}(\mathbb{F}_q^{n\times m})$ such that $\mathcal{M}\simeq \mathcal{M}_{\mathcal{C}}$.

\section{Secret Sharing With Rank-Metric Codes}

\textcolor{black}{We begin this section by motivating how secret-sharing schemes can be defined using rank-metric codes, formalizing these notions in the subsequent subsections. Let $P_0,P\in \mathcal{L}(\mathbb{F}_q^n)$ such that $P_0 \oplus P = \mathbb{F}_q^n$, and let $\mathcal{C}\leq \mathcal{L}(\mathbb{F}_q^{n\times m})$ be a rank-metric code with $\mathcal{C}(P_0^\perp)\neq \mathcal{C}$. For any $V\in \mathcal{L}(\mathbb{F}_q^n)$ let $G_V$ be a fixed public generator matrix of full rank with $\mathrm{rowsp}(G_V) = V$. The dealer corresponds to $P_0$ and the player space corresponds to $P$. The players are the $1$-dimensional spaces of $P$, resulting in a total of $\frac{q^{\dim P}-1}{q-1}$ players. The dealer selects a secret $S\in G_{P_0}\mathcal{C}$ and uniformly chooses $X\in \mathcal{C}$ such that $G_{P_0}X = S$. The condition $\mathcal{C}(P_0^\perp)\neq \mathcal{C}$ ensures $\dim G_{P_0}\mathcal{C} \neq 0$. Each player $p_i\in \mathcal{L}(P)$ receives $G_{p_i}X$ from the dealer. A coalition $V \in \mathcal{L}(P)$ can reconstruct $G_{P_0}X$ using $G_V X$ if $G_V Y_1 = G_V Y_2$ implies $G_{P_0}Y_1 = G_{P_0}Y_2$ for any $Y_1,Y_2\in \mathcal{C}$. By linearity, this is equivalent to $G_V Z = 0\Rightarrow G_{P_0}Z = 0$ for all $Z \in \mathcal{C}$. This, in turn, is equivalent to $\rho_{\mathcal{C}}(V) = \rho_{\mathcal{C}}(V + P_0)$. Furthermore, for $V\in \mathcal{L}(E)$, the number of feasible $G_{P_0}X$ given $G_{V}X$ is $q^{m \rho(P_0 \,|\, V)}$. Thus, the probability that the coalition $V$ correctly determines the secret by {\em uniformly random guessing} is $q^{-m \rho_{\mathcal{C}}(P_0 \,|\, V)}$.}

\begin{example} \label{exmp}Consider $\mathcal{C} = \langle M_1,M_2,M_3,M_4,M_5,M_6\rangle_{\mathbb{F}_2}$ where 
    \begin{align*}M_1 =
    \begin{bsmallmatrix}
1 & 0 & 0 & 0 & 0 & 0 \\
0 & 1 & 0 & 0 & 1 & 0 \\
0 & 0 & 0 & 0 & 0 & 1 \\
0 & 1 & 0 & 1 & 0 & 1
\end{bsmallmatrix}, M_2 =
\begin{bsmallmatrix}
0 & 1 & 0 & 0 & 0 & 0 \\
1 & 0 & 0 & 1 & 1 & 1 \\
0 & 0 & 1 & 1 & 0 & 1 \\
1 & 1 & 1 & 1 & 1 & 1
\end{bsmallmatrix}, M_3=
\begin{bsmallmatrix}
0 & 0 & 1 & 0 & 0 & 0 \\
1 & 1 & 0 & 1 & 0 & 1 \\
0 & 1 & 0 & 1 & 1 & 0 \\
0 & 0 & 0 & 0 & 1 & 0
\end{bsmallmatrix},
\end{align*}
\begin{align*} 
M_4 =
\begin{bsmallmatrix}
0 & 0 & 0 & 1 & 0 & 0 \\
1 & 1 & 0 & 0 & 0 & 0 \\
1 & 1 & 1 & 0 & 0 & 1 \\
1 & 0 & 1 & 1 & 1 & 1
\end{bsmallmatrix}, M_5=
\begin{bsmallmatrix}
0 & 0 & 0 & 0 & 1 & 0 \\
0 & 1 & 1 & 1 & 1 & 1 \\
0 & 0 & 1 & 1 & 0 & 0 \\
1 & 0 & 1 & 1 & 0 & 0
\end{bsmallmatrix}, M_6=
\begin{bsmallmatrix}
0 & 0 & 0 & 0 & 0 & 1 \\
0 & 1 & 0 & 1 & 1 & 0 \\
0 & 1 & 1 & 1 & 0 & 1 \\
0 & 0 & 1 & 0 & 1 & 0
\end{bsmallmatrix}.
\end{align*}

Let $P_0 = \langle \begin{bsmallmatrix}
1 & 0 & 0 & 1
\end{bsmallmatrix}\rangle_{\mathbb{F}_2}$ and $P = \langle \mathbf{e}_2,\mathbf{e}_3,\mathbf{e}_4\rangle_{\mathbb{F}_2}$. Then $P_0 \oplus P = \mathbb{F}_2^4$ and $\mathcal{C}(P_0^\perp) \neq \mathcal{C}$. As the scheme is over $\mathbb{F}_2$, there is only a single choice for generator matrices $G_V$ for any $V\in \mathcal{L}(\mathbb{F}_2^4)$ with $\dim V = 1$, and the dealer need not transmit these. The dealer chooses a secret $S\in G_{P_0}\mathcal{C}$ and $X\in \mathcal{C}$ such that
\[
X =
\begin{bsmallmatrix}
    0 & 1 & 1 & 1 & 0 & 0 \\
1 & 0 & 0 & 0 & 1 & 0 \\
1 & 0 & 0 & 0 & 1 & 0 \\
0 & 1 & 0 & 0 & 1 & 0
\end{bsmallmatrix}, \: G_{P_0} X = S= \begin{bsmallmatrix}
    0 & 0 & 1 & 1 & 1 & 0
\end{bsmallmatrix}.
\]
The players corresponding to $p_1 = \langle \begin{bsmallmatrix}
    0 & 0 & 1 & 0 
\end{bsmallmatrix}\rangle_{\mathbb{F}_2}$ and $p_2 =\langle \begin{bsmallmatrix}
    0 & 1 & 1 & 0 
\end{bsmallmatrix}\rangle_{\mathbb{F}_2}$ then receive their corresponding shares $G_{p_1}X = \begin{bsmallmatrix}
    1 & 0 & 0 & 0 & 1 & 0
\end{bsmallmatrix}$ and $G_{p_2}X = \begin{bsmallmatrix}
    1 & 1 & 0 & 0 & 0 & 0 
\end{bsmallmatrix}$. If $p_i$ then wishes to reconstruct the secret they determine the space
\[
    \Omega(G_{p_i},X) = \{G_{P_0}Y  \,|\, Y\in \mathcal{C}\colon G_{p_i} Y = G_{p_i}X \}.
\]
If $|\Omega(p_i,X)| = 1$ the player will have uniquely reconstructed the secret as $G_{P_0}X \in \Omega(p_i,X)$ for all $p_i$ and $X$. However, neither $p_1$ or $p_2$ are able to reconstruct the secret as $|\Omega(G_{p_1},X)| = 2$ and $|\Omega(G_{p_2},X)|=4$, meaning that both players gain partial information of the secret as $|G_{P_0} \mathcal{C}| = 32$. This is also seen by $\rho_{\mathcal{C}}$ as $\rho_{\mathcal{C}}(p_1) = 5/6$ and $\rho_{\mathcal{C}}(p_2) = 2/3$ while $\rho_{\mathcal{C}}(p_1 + P_0) = \rho_{\mathcal{C}}(p_2 + P_0) =1$. However, $p_1$ is stronger in the sense that they obtain more information than $p_2$, also seen by $\rho_{\mathcal{C}}(p_1) > \rho_{\mathcal{C}}(p_2)$. If $p_1$ and $p_2$ collude, they pool their shares to obtain
 \[
G_{p_1+p_2} = \begin{bsmallmatrix}
    G_{p_1} \\ G_{p_2}
\end{bsmallmatrix} = \begin{bsmallmatrix}
    1 & 0 & 0 & 0 & 1 & 0 \\
    1 & 1 & 0 & 0 & 0 & 0
\end{bsmallmatrix}.
\]
As $|\Omega(G_{p_1+p_2},X)| = 1$, the coalition $p_1+p_2$ is able to reconstruct or, equivalently, $\rho_{\mathcal{C}}(P_0 \,|\, p_1 + p_2) = 0$.
\end{example}

\subsection{Access Structures on Subspaces}
In the following, we generalize access structures on finite sets to that of vector spaces. For sets $\Gamma, \mathcal{A} \subseteq \mathcal{L}(E)$, we naturally extend the notions of monotonicity and anti-monotonicity of finite sets to $\mathcal{L}(E)$ with respect to subspace inclusion. Lastly, two monotone sets, or two anti-monotone sets, $H_1, H_2$, are \emph{equivalent} if there exists a $\mathbb{F}_q$-linear monotone bijection $f: H_1 \rightarrow H_2$ with a monotone inverse.
\begin{definition}
Let $\Gamma,\mathcal{A}\subseteq \mathcal{L}(E)$ such that $\Gamma \cap \mathcal{A} = \varnothing$, $\Gamma$ is monotone and $\mathcal{A}$ is anti-monotone. Then $\mathbf{S} =(\Gamma,\mathcal{A})$ is an access structure on $\mathcal{L}(E)$.
\end{definition}

Henceforth, an access structure will always be on subspaces unless specified.

\begin{definition}
Let $\mathbf{S} = (\Gamma,\mathcal{A})$ be an access structure on $\mathcal{L}(E)$.
\begin{enumerate}
    \item $\mathbf{S}$ is degenerate if $\Gamma = \varnothing$ or $\mathcal{A}=\varnothing$.
    \item $\Gamma_{\min} = \{V\in \Gamma\,|\, W< V \Rightarrow W\notin \Gamma\}$ is the minimal reconstructing structure.
    \item $\mathbf{S}$ is perfect if $\Gamma \cup \mathcal{A}= \mathcal{L}(E)$. 
    \item The minimum gap of $\mathbf{S}$ is $$g(\mathbf{S}) = \min \{ \dim(V/W) \,|\, (V,W)\in \Gamma \times \mathcal{A} \colon W \leq V\}.$$
    \item $\mathbf{S}$ is a $k$-threshold structure if for some $k\in \mathbb{N}$ then $$\Gamma = \{V\in \mathcal{L}(E) \,|\, \dim V \geq k\}.$$
\end{enumerate}
\end{definition}
A perfect access structure $\mathbf{S}$ satisfies $g(\mathbf{S})=1$ and it is entirely determined by $\Gamma_{\min}$.

\begin{definition}
Let $\mathbf{S}_i = (\Gamma_i,\mathcal{A}_i)$ be access structures on $\mathcal{L}(E_i)$ for $i=1,2$. If there exists an $\mathbb{F}_q$-linear monotone bijection $f \colon \mathcal{L}(E_1) \rightarrow \mathcal{L}(E_2)$ with a monotone inverse such that $f(\Gamma_1) = \Gamma_2$ and $f(\mathcal{A}_1) = \mathcal{A}_2$, then $\mathbf{S}_1$ and $\mathbf{S}_2$ are equivalent, denoted by $\mathbf{S}_1\simeq \mathbf{S}_2$.
\end{definition}

We now define minors on monotone and anti-monotone sets. These definitions are then naturally extended to minors on access structures. The mutual duality between restriction and contraction is given in Theorem \ref{thm:minors}, analogous to the case of $q$-polymatroids \cite{BYRNE2024105799,DBLP:journals/corr/abs-2104-06570}.

\begin{definition} Let $H\subseteq \mathcal{L}(E)$ be a monotone or anti-monotone set and $Z\in \mathcal{L}(E)$. Then 
\begin{enumerate}
    \item $\overline{H} = \{V \in \mathcal{L}(E)\,|\, V\notin H\}$ is the complement.
    \item $H^*=\{ V\in \mathcal{L}(E) \,|\, V^{\perp}\in H \}$ is the dual.
    \item $H|_Z = \{ V \in \mathcal{L}(Z)\,|\, V\in H\}$ is the restriction to $Z$.
    \item $H/Z = \{V \in \mathcal{L}(E/Z) \,|\, \pi^{-1}(V)\in H\}$ is the contraction by $Z$, where $\pi\colon E\rightarrow E/Z$ is the canonical quotient map.
\end{enumerate}
\end{definition}
The dual of a monotone set is clearly anti-monotone, and vice-versa. Similarly, monotonicity and anti-monotonicity is preserved under restriction and contraction.

\begin{definition} Let $\mathbf{S}=(\Gamma,\mathcal{A})$ be an access structure on $\mathcal{L}(E)$ and $Z\in \mathcal{L}(E)$. Then
\begin{enumerate}
    \item $\mathbf{S}^* = (\mathcal{A}^*,\Gamma^*)$ is the dual access structure.
    \item $\mathbf{S} |_Z = (\Gamma|_Z,\mathcal{A}|_Z)$ is the restriction to $Z$.
    \item $\mathbf{S}/Z = (\Gamma/Z,\mathcal{A}/Z)$ is the contraction by $Z$.
\end{enumerate}
\end{definition}

Note that when dualizing $\mathbf{S}^*$, the dualization of $\Gamma$ and $\mathcal{A}$ is with respect to the same non-degenerate symmetric bilinear form. Additionally, the equivalence class of $\mathbf{S}^*$ is independent of the choice of non-degenerate symmetric bilinear form.

 Dualization and taking minors of access structures also yield access structures. In particular, $\mathbf{S}|_Z$ is an access structure on $\mathcal{L}(Z)$ and $\mathbf{S}/Z$ is an access structure on $\mathcal{L}(E/Z)$. The contraction by $Z$ essentially describes how the scheme changes when the information held by $Z$ is given to all players. Lastly, the minors of perfect access structures are also perfect.

\begin{theorem} \label{thm:minors}
Let $\mathbf{S}$ be an access structure on $\mathcal{L}(E)$. Then
\[
    (\mathbf{S}/Z)^* \simeq \mathbf{S}^*|_{Z^{\perp}} \quad \text{and} \quad(\mathbf{S}| _{Z^{\perp}})^* \simeq \mathbf{S}^*/Z.
\]
\end{theorem}
\begin{proof} We prove the first equivalence as the second then follows by duality. Let $\pi\colon E \rightarrow E/Z$ denote the canonical quotient map. It is possible to choose bilinear forms such that $Z\oplus Z^{\perp} = E$ and there exists a monotone isomorphism $\phi\colon E/Z\rightarrow Z^{\perp}$ with a monotone inverse satisfying $\pi^{-1}(V)^{\perp} = \phi(V)^{\perp} \cap Z^{\perp}$ for all $V\in E/Z$, see \cite[Theorem 5.3]{DBLP:journals/corr/abs-2104-06570}. In particular,
\begin{align*}
    (\Gamma/Z)^* &= \{V\in \mathcal{L}(E/Z)\,|\, \pi^{-1}(V^{\perp E/Z}) \in \Gamma \},\\
    \Gamma^*|_{Z^{\perp}} &= \{ W\in \mathcal{L}(Z^{\perp}) \, |\, W^{\perp}\in \Gamma\},
\end{align*}
and similarly for $(\mathcal{A}/Z)^*$ and $\mathcal{A}^*|_{Z^{\perp}}$. Consider now the maps $\sigma\colon \mathcal{L}(E/Z) \rightarrow \mathcal{L}(Z^{\perp})$ such that $V\mapsto \pi^{-1}(V^{\perp E/Z})^{\perp}$
and $\tau \colon \mathcal{L}(Z^{\perp})\rightarrow \mathcal{L}(E/Z)$ such that $W\mapsto (\phi^{-1}((W+Z)^{\perp}))^{\perp E/Z}$. One may then confirm that these are well-defined, monotone, $\mathbb{F}_q$-linear, and mutually inverse. The proof then follows as $\sigma((\Gamma/Z)^*) = \Gamma^*|_{Z^{\perp}}$ and $\sigma((\mathcal{A}/Z)^*) = \mathcal{A}^*|_{Z^{\perp}}$.
\end{proof}

\subsection{$q$-Polymatroid Ports}
The following is the generalization of matroid ports to that of $q$-polymatroids, see \cite{farre_padro2007}.

\begin{definition}
Let $\mathcal{M}=(E,\rho)$ be a $q$-polymatroid and $P_0,P\in \mathcal{L}(E)$ with $P_0 \oplus P = E$. If $\rho(P_0) > 0$ then $\mathbf{S}_{P_0,P}(\mathcal{M}) = (\Gamma, \mathcal{A})$, where
\begin{align*}
    \Gamma &\coloneq \{V\in \mathcal{L}(P) \, |\, \rho(P_0 \,|\, V ) = 0\}, \\
    \mathcal{A} &\coloneqq \{W\in \mathcal{L}(P) \,|\, \rho(P_0\,|\, W ) = \rho_{\mathcal{C}}(P_0)\},
\end{align*}
is a generalized $q$-polymatroid port.
\end{definition}
Proposition \ref{prop:downsetupset} shows that a generalized $q$-polymatroid port $\mathbf{S}_{P_0,P}(\mathcal{M})$ is an access structure, where $\Gamma$ is the reconstructing structure and $\mathcal{A}$ is the privacy structure of the port.  Furthermore, if $ \mathcal{M}$ is a $q$-matroid, then $\mathbf{S}_{P_0,P}(\mathcal{M})$ is a generalized $q$-matroid port. If $\dim P_0 = 1$, it is a $q$-polymatroid port. If both conditions are satisfied, it is a $q$-matroid port. 

\begin{proposition}\label{prop:downsetupset} Let $\mathbf{S}_{P_0,P}(\mathcal{M})$ be a generalized $q$-polymatroid port. Then $\mathbf{S}_{P_0,P}(\mathcal{M})$ is an access structure on $\mathcal{L}(P)$. Additionally, $\mathbf{S}_{P_0,P}(\mathcal{M})$ is degenerate if and only if $\Gamma = \varnothing$.
\end{proposition}
\begin{proof}
We prove the monotonicity of $\Gamma$, as the anti-monotonicity of $\mathcal{A}$ then follows by duality. Let $V \in \Gamma$ and suppose $V' \in \mathcal{L}(P)$ such that $V \leq V'$ but $V' \notin \Gamma$. Then,
\[
    \rho(V') - \rho(V+P_0) < \rho(V'+P_0) - \rho(V+P_0 ) \leq \rho(V') - \rho(V)
\]
which contradicts the assumption that $V \in \Gamma$.
\end{proof}

\begin{example} \label{exmp2} Consider the code $\mathcal{C}$ of Example \ref{exmp}. Then $\mathbf{S}_{P_0,P}(\mathcal{M}_\mathcal{C})$ is a $q$-polymatroid port, and the access structure is characterized by $\mathcal{A}=\{0\}$, $\Gamma_{\min} = \{\langle \mathbf{e}_2, \mathbf{e}_3 \rangle_{\mathbb{F}_2}
,\langle \mathbf{e}_2+\mathbf{e}_4,\mathbf{e}_3+\mathbf{e}_4 \rangle_{\mathbb{F}_2}
, \langle \mathbf{e}_2+\mathbf{e}_4,\mathbf{e}_3 \rangle_{\mathbb{F}_2},\langle\mathbf{e}_2, \mathbf{e}_3+\mathbf{e}_4 \rangle_{\mathbb{F}_2}
, \langle \mathbf{e}_4 \rangle_{\mathbb{F}_2}
    \}$, and $\overline{\Gamma}\setminus \mathcal{A} = \{\langle \mathbf{e}_3 \rangle_{\mathbb{F}_2}, \langle \mathbf{e}_2+\mathbf{e}_4 \rangle_{\mathbb{F}_2},
\langle \mathbf{e}_3+\mathbf{e}_4 \rangle_{\mathbb{F}_2}, 
\langle \mathbf{e}_2 \rangle_{\mathbb{F}_2}
, \langle \mathbf{e}_2+\mathbf{e}_3 \rangle_{\mathbb{F}_2},
\langle \mathbf{e}_2+\mathbf{e}_3+\mathbf{e}_4 \rangle_{\mathbb{F}_2}
 \}.$
\end{example}

We now generalize the notion of information ratio to generalized $q$-polymatroid ports and demonstrate how it relates to the minimum gap, entirely analogous to the case of classical polymatroid ports \cite{cryptoeprint:2012/595}. As such, the information ratio is a measure of the size of the shares compared to the size of the secret.

\begin{definition}
    Let $\mathbf{S} = \mathbf{S}_{P_0,P}(\mathcal{M})$ be a generalized $q$-polymatroid port and $\mathcal{P}(P)$ denote the $1$-dimensional spaces of $\mathcal{L}(P)$. Then $\sigma(\mathbf{S}) = \max_{p \in \mathcal{P}(P)}\{ \rho(p)\}/\rho(P_0)$ is the information ratio of $\mathbf{S}$.
\end{definition}

\begin{proposition}
Let $\mathbf{S} = \mathbf{S}_{P_0,P}(\mathcal{M})$ be a generalized $q$-polymatroid port. Then $\sigma(\mathbf{S})\geq g(\mathbf{S})^{-1}$.
\end{proposition}
\begin{proof}
Let $V\in \Gamma$ and $W\in \mathcal{A}$ with $W\leq V$ and fix some $Y\in \mathcal{L}(P)$ such that $W\oplus Y = V$. Then $\rho(P_0)\leq \rho(Y) \leq \dim Y \max_{p\in \mathcal{P}(P)}\{\rho(p)\}$ and the result follows by taking the minimum over all such pairs $(V,W)$.
\end{proof}

Proposition \ref{prop:minors} shows how the minors of an access structure from a generalized $q$-polymatroid port relates to the ports of the minors of the underlying $q$-polymatroid. 

\begin{proposition}\label{prop:minors}
    Let $\mathbf{S}_{P_0,P}(\mathcal{M})$ be a generalized $q$-polymatroid port and $Z\in \mathcal{L}(P)$. Then \linebreak $\mathbf{S}_{P_0,Z}(\mathcal{M}|_{P_0 + Z})$ is a generalized $q$-polymatroid port and
    $$\mathbf{S}_{P_0,P}(\mathcal{M})|_Z=\mathbf{S}_{P_0,Z}(\mathcal{M}|_{P_0 + Z}).$$
    Furthermore, if $Z\in \overline{\Gamma}$ then $\mathbf{S}_{\pi(P_0),\pi(P)}(\mathcal{M}/Z)$ is a generalized $q$-polymatroid port and
    $$\mathbf{S}_{P_0,P}(\mathcal{M})/Z = \mathbf{S}_{\pi(P_0),\pi(P)}(\mathcal{M}/Z),$$
    where $\pi\colon E\rightarrow E/Z$ is the canonical quotient map.
\end{proposition}
\begin{proof}
The first statement follows directly as $\rho|_{P_0+Z}(V) = \rho(V)$ for all $V\in \mathcal{L}(P_0+Z)$. The second statement holds as $\rho_{E/Z}(\pi(P_0)) = \rho(P_0 \,|\, Z)$ and $\rho_{E/Z}(\pi(P_0)\,|\, W) = \rho(P_0\,|\, \pi^{-1}(W))$ for all $W\in \mathcal{L}(P/Z)$.
\end{proof}

Proposition \ref{prop:equivports} shows how equivalent $q$-polymatroids yield equivalent ports under appropriate choices of $P_0$ and $P$ in the following sense. The proof follows directly by Definition \ref{def:equivalentmatroids}.

\begin{proposition} \label{prop:equivports}
    Let $\mathbf{S}_{P_0,P}(\mathcal{M})$ be a generalized $q$-polymatroid port. If $\mathcal{M}\simeq \mathcal{M'}$ under $\phi$ then $\mathbf{S}_{\phi(P_0),\phi(P)}(\mathcal{M}')$ is a generalized $q$-polymatroid port and $$\mathbf{S}_{P_0,P}(\mathcal{M}) \simeq \mathbf{S}_{\phi(P_0),\phi(P)}(\mathcal{M}').$$
\end{proposition}

Dualizing the port or the underlying $q$-polymatroid yields equivalent ports in the following sense.

\begin{proposition}
    Let $\mathbf{S}_{P_0,P}(\mathcal{M})$ be a non-degenerate generalized $q$-polymatroid port with $\rho(P_0) = \dim P_0$. Then
    \[
        \mathbf{S}_{P_0,P}(\mathcal{M})^* \simeq \mathbf{S}_{P^{\perp},P_0\textcolor{white}{)}\hspace{-1.5mm}^{\perp}}(\mathcal{M}^*).
    \]
\end{proposition}

\begin{proof}
    Let $V\in \mathcal{L}(P)$ and $\bot = \bot E$. Then
    $$
        \rho^*(P^\perp\,|\, V^\perp \cap P_0^\perp) = \dim P^\perp - \rho(P_0\,|\, V).
    $$ The map $\tau \colon \mathcal{L}(P)\rightarrow \mathcal{L}(P_0^\perp)$ such that $V\mapsto V^\perp\cap P_0^\perp$ is then an anti-monotone bijection and the result follows as
    \[
    V\in \Gamma^* \Leftrightarrow V^{\perp P}\in \Gamma \Leftrightarrow \rho^*(P^\perp\,|\, \tau(V^{\perp P})) = \rho^*(P^\perp)
    \]
    and
    \[
    W\in \mathcal{A}^* \Leftrightarrow W^{\perp P}\in \mathcal{A} \Leftrightarrow \rho^*(P^\perp\,|\, \tau(W^{\perp P})) =0.
   \qedhere \]
    \end{proof}

An access structure realized by a $q$-matroid port is perfect, and its minimal access structure can be characterised by the minors of the underlying $q$-matroid.
\begin{theorem}\label{thm:qmatroidports} Let $\mathbf{S}_{P_0,P}(\mathcal{M})$ be a $q$-matroid port. Then $\mathbf{S}_{P_0,P}(\mathcal{M})$ is perfect and \begin{align*}
    \Gamma_{\min} = \{V\in \mathcal{L}(P) \,|\, &\textcolor{black}{V \text{ is a basis in } \mathcal{M}|_{V+P_0} \;} \textcolor{black}{   and }\; W+P_0 \\ &\textcolor{black}{ \text{ is independent for all } W<V}\}.
    \end{align*}
\end{theorem}
\begin{proof}
$\mathbf{S}_{P_0,P}(\mathcal{M})$ is perfect as $0\leq \rho(P_0\,|\,V)\leq \rho(P_0) =1$ for any $V\in \mathcal{L}(P)$. \textcolor{black}{The secondary statement follows easily by definitions of restrictions and bases of $q$-matroids.}
\end{proof}
Whereas linear codes always induce matroids, ${\mathbb F}_q$-linear rank-metric codes do not necessarily induce $q$-matroids. Consequently, not all ports induced by rank-metric codes are perfect, as illustrated in Example \ref{exmp}. Theorem \ref{thm:qmatroidports} is the generalization of a well-known result on matroid ports, see \cite{farre_padro2007}. However, it differs in the sense that the minimal access structure of a $q$-matroid port is not the set of $V\in \mathcal{L}(P)$ where $V+P_0$ is a circuit in the underlying $q$-matroid. It is indeed sufficient that if $V\in\mathcal{L}(P)$ such that $V+P_0$ is a circuit, then $V\in \Gamma_{\min}$. It is however not necessary, as Example \ref{exmp:circuitminimal} shows.
\begin{example} \label{exmp:circuitminimal}
Consider $\mathcal{C} = \langle M_1,M_2,M_3,M_4\rangle_{\mathbb{F}_2}$ where
\[ M_1 =
\begin{bsmallmatrix}
    1 & 0 \\
    0 & 0 \\
    1 & 1 \\
    0 & 1 \\
\end{bsmallmatrix},\:
M_2 =
\begin{bsmallmatrix}
    0 & 1 \\
    0 & 0 \\
    1 & 0 \\
    1 & 1 \\
\end{bsmallmatrix},\:
M_3 =
\begin{bsmallmatrix}
    0 & 0 \\
    1 & 0 \\
    1 & 0 \\
    1 & 0 \\
\end{bsmallmatrix},\:
M_4 =
\begin{bsmallmatrix}
    0 & 0 \\
    0 & 1 \\
    0 & 1 \\
    0 & 1 \\
\end{bsmallmatrix}.
\]
Let $P_0 = \langle \mathbf{e}_1 \rangle_{\mathbb{F}_2}$ and $P = \langle \mathbf{e}_2,\mathbf{e}_3,\mathbf{e}_4\rangle_{\mathbb{F}_2} $. Then $\rho_{\mathcal{C}}(P_0)>0$ so $\mathbf{S}_{P_0,P}(\mathcal{M}_{\mathcal{C}})$ is a perfect $q$-matroid port by Theorem \ref{thm:qmatroidports}. Furthermore, 
\[
    \Gamma_{\min} = \{\langle  \mathbf{e}_2+\mathbf{e}_4 \rangle_{\mathbb{F}_2},
    \langle  \mathbf{e}_2+\mathbf{e}_3 \rangle_{\mathbb{F}_2},
\langle  \mathbf{e}_3+\mathbf{e}_4 \rangle_{\mathbb{F}_2}
\},
\]
but the circuits of $\mathcal{M}_{\mathcal{C}}$ are $\{\langle\mathbf{e}_1, \mathbf{e}_2+\mathbf{e}_3 \rangle_{\mathbb{F}_2},\langle \mathbf{e}_2+\mathbf{e}_4,\mathbf{e}_3+\mathbf{e}_4 \rangle_{\mathbb{F}_2},$ $\langle\mathbf{e}_1,\mathbf{e}_2+\mathbf{e}_4 \rangle_{\mathbb{F}_2}$ $\langle \mathbf{e}_1+\mathbf{e}_2+\mathbf{e}_4,\mathbf{e}_3+\mathbf{e}_4 \rangle_{\mathbb{F}_2},\langle\mathbf{e}_1+\mathbf{e}_3+\mathbf{e}_4 \rangle_{\mathbb{F}_2}\}$ so $v =\langle \mathbf{e}_3+\mathbf{e}_4 \rangle_{\mathbb{F}_2} + P_0$ is not a circuit as $\langle \mathbf{e}_1+\mathbf{e}_3+\mathbf{e}_4 \rangle_{\mathbb{F}_2}\leq v$.
\end{example}

In classical secret sharing, maximum distance separable (MDS) codes give rise to perfect threshold secret sharing schemes. We conclude this subsection by demonstrating that this is also true for MRD codes under our definition of threshold structures on the poset of subspaces. 

\begin{proposition}
Let $\mathbf{S}_{P_0,P}(\mathcal{M}_{\mathcal{C}})$ be a generalized $q$-polymatroid port. Suppose $\mathcal{C}$ is an MRD code with $\dim \mathcal{C} = k$. Then $\mathbf{S}_{P_0,P}(\mathcal{M}_{\mathcal{C}})$ is a $\frac{k}{m}$-threshold structure.
\end{proposition}
\begin{proof}
This follows from the induced rank function of an MRD code being of the form $\rho_\mathcal{C}(V) = \min\{\dim V, \frac{\dim (\mathcal{C})}{m} \}$ for all $V\in \mathcal{L}(E)$ \cite[Cor. 6.6]{Gorla_2019}.
\end{proof}

\subsection{The entropy function of representable $q$-polymatroids}

It is natural to consider an information theoretic description of the rank function of a generalized $q$-polymatroid port.
Theorem \ref{thm:entropy} and Lemma \ref{lem:entropy} gives a $q$-analogue to the classical {\em entropic polymatroids} \cite{10.1007/s10623-020-00811-1,Abbe_2019}, defined for representable $q$-polymatroids.

Let $H(\mathbf{X}_1,\ldots,\mathbf{X}_n)$ denote the entropy of a set of discrete random variables $\mathbf{X}_1,\ldots,\mathbf{X}_n$. Additionally, the conditional entropy of $\mathbf{X}_1$ given $\mathbf{X}_2$ is denoted $H(\mathbf{X}_1\,|\, \mathbf{X}_2) = H(\mathbf{X}_1,\mathbf{X}_2) - H(\mathbf{X}_2)$.

\begin{theorem} \label{thm:entropy}
Let $\mathcal{C}\in \mathcal{L}(\mathbb{F}_q^{n\times m})$ be a rank-metric code and $V\in \mathcal{L}(\mathbb{F}_q^n)$. Let $\pi_V \colon \mathcal{C}\rightarrow \mathcal{C}/\mathcal{C}(V^{\perp \mathbb{F}_q^n})$ denote the canonical quotient map. Let $\mathbf{Z}_V$ be the random variable on $\mathcal{C}/\mathcal{C}(V^{\perp\mathbb{F}_q^n})$ induced by the uniform distribution on $\mathcal{C}$ and $\pi_V$. Then
$$H(\mathbf{Z}_V)= m\log(q)\rho_{\mathcal{C}}(V).$$
\end{theorem}
\begin{proof}
Let $\bot = \bot \mathbb{F}_q^n$. Then
    \begin{align*}
     H(\mathbf{Z}_V) &= -\:\:\sum_{\mathclap{X\in \mathcal{C}/\mathcal{C}(V^\perp)}}\:\: \mathbb{P}(X) \log (\mathbb{P}(X)) \\
    &= -\:\:\sum_{\mathclap{X\in \mathcal{C}/\mathcal{C}(V^\perp)}}\:\: \frac{q^{\dim \ker \pi_V}}{q^{\dim (\mathcal{C})}} \log \left( \frac{q^{\dim \ker \pi_V}}{q^{\dim (\mathcal{C})}}\right)\\
    &= -q^{\dim \mathcal{C}/\mathcal{C}(V^\perp)}\frac{q^{\dim \ker \pi_V}}{q^{\dim (\mathcal{C})}} \log \left( \frac{q^{\dim \ker \pi_V}}{q^{\dim (\mathcal{C})}}\right) \\
    &= (\dim \mathcal{C} - \dim \mathcal{C}(V^\perp)) \log(q)\\ 
    &= m\log(q)\rho_C(V).\qedhere
\end{align*}
\end{proof}
\begin{lemma} \label{lem:entropy} Assume the hypothesis of Theorem \ref{thm:entropy}. Let $V_1,$ $\ldots,V_\ell \in \mathcal{L}(\mathbb{F}_q^n)$. Then
$$
H(\mathbf{Z}_{V_1},\ldots,\mathbf{Z}_{V_\ell}) = H( \mathbf{Z}_{V_1+ \cdots + V_\ell}).$$
\end{lemma}
\begin{proof}
Let $\bot = \bot \mathbb{F}_q^n$. Let $\sigma$ be the map $\sigma \colon \mathcal{C} \rightarrow \prod^{\ell}_{i=1}\mathcal{C}/\mathcal{C}(V_i^\perp)$ defined by $$\sigma(X)= (\pi_{V_1}(X),\ldots,\pi_{V_\ell}(X))$$ for all $X \in \mathcal{C}$. This map is clearly $\mathbb{F}_q$-linear and satisfies $\ker \sigma = \cap^{\ell}_{i=1}\ker \pi_{V_i} = \ker \pi_{\sum_{i=1}^\ell V_i}$. The result then follows by the definition of joint entropy as $\mathrm{rk}\, \sigma = \mathrm{rk}\, \pi_{\sum^{\ell}_{i=1} V_i}$ and for any $Y=(Y_1,\ldots,Y_\ell)\in \mathrm{im}\, \sigma$ it holds that $\mathbb{P}(Y) = q^{\dim \ker \sigma - \dim \mathcal{C}}$.
\end{proof}
From Theorem \ref{thm:entropy} and Lemma \ref{lem:entropy} we immediately derive Corollary \ref{cor:entropy} showing that the conditional rank of representable $q$-polymatroids is well-defined in terms of conditional entropy.
\begin{corollary}\label{cor:entropy} Assume the hypothesis of Theorem \ref{thm:entropy}. Let  $V,W\in \mathcal{L}(\mathbb{F}_q^n)$. Then $$H(\mathbf{Z}_{W} \,|\, \mathbf{Z}_{V})=m \log(q)\rho_{\mathcal{C}}(W\, |\, V).$$
\end{corollary}

\section{Conclusion}
We have generalized the concepts of secret sharing and matroid ports to $q$-polymatroids and shown how rank-metric codes give rise to secret sharing schemes within this framework. One possible application of this theory is to random linear wiretap networks \cite{Onsimilarities,unifying}, which can be treated as a special case of the theory developed in this paper. The link between ports of representable $q$-polymatroids and wiretap coset coding schemes for random linear network coding will be made more explicit in the full version of this paper. This in turn could be useful for private information retrieval over random linear networks \cite{tajeddine-network}.

\renewcommand{\baselinestretch}{0.98}

\bibliographystyle{IEEEtran}
\bibliography{bibliography}

\end{document}